\documentclass[12pt]{article}      
\usepackage{amsmath,amssymb,amsthm, amscd, graphicx, verbatim, setspace} 
\usepackage{setspace}

\theoremstyle{plain}
\newtheorem{theo}{Theorem}
\newtheorem{lem}[theo]{Lemma}
\newtheorem{corollary}[theo]{Corollary}

\title{Sequences of formation width $4$ and alternation length $5$}
\date{}
\author{Jesse Geneson \qquad Peter Tian\\
\small Department of Mathematics\\[-0.8ex]
\small MIT\\[-0.8ex]
\small Massachusetts, U.S.A.\\
\small\tt geneson@math.mit.edu\\
\small\tt tianpeter15@yahoo.com
}
\begin{document}
\maketitle

\begin{abstract}
Sequence pattern avoidance is a central topic in combinatorics. A sequence $s$ {\it contains} a sequence $u$ if some subsequence of $s$ can be changed into $u$ by a one-to-one renaming of its letters. If $s$ does not contain $u$, then $s$ {\it avoids} $u$. A widely studied extremal function related to pattern avoidance is $Ex(u, n)$, the maximum length of an $n$-letter sequence that avoids $u$ and has every $r$ consecutive letters pairwise distinct, where $r$ is the number of distinct letters in $u$. 

We bound $Ex(u, n)$ using the formation width function, $fw(u)$, which is the minimum $s$ for which there exists $r$ such that any concatenation of $s$ permutations, each on the same $r$ letters, contains $u$. In particular, we identify every sequence $u$ such that $fw(u)=4$ and $u$ contains $ababa$. The significance of this result lies in its implication that, for every such sequence $u$, we have $Ex(u, n) = \Theta(n \alpha(n))$, where $\alpha(n)$ denotes the incredibly slow-growing inverse Ackermann function. We have thus identified the extremal function of many infinite classes of previously unidentified sequences.

  \bigskip\noindent \textbf{Keywords:} alternations, formations, generalized Davenport-Schinzel sequences, inverse Ackermann functions, permutations
\end{abstract}

\section{Introduction}

Pattern avoidance in sequences is a widely applicable topic in combinatorics. The field was inititated in 1965 by Davenport and Schinzel \cite{3}, who introduced sequences avoiding certain patterns to study linear differential equations. Specifically they introduced Davenport-Schinzel Sequences, which avoid alternations of two letters. More precisely, $u_1u_2 \cdots u_m$ is a Davenport-Schinzel sequence of order $s$ if it satisfies 

\begin{itemize}

\item $u_i \not= u_{i+1}$ for each index $i<m$
\item There do not exist indices $1 \le i_1 < \cdots < i_{s+2} \le m$ such that $u_{i_1}=u_{i_3}=\cdots=a$ and $u_{i_2}=u_{i_4}=\cdots=b$, for some integers $a \not= b$.

\end{itemize}

Upper bounds on the lengths of Davenport-Schinzel sequences have been used to bound the complexity of lower envelopes of sets of polynomials of limited degree \cite{3} and the complexity of faces in arrangements of arcs with a limited number of crossings \cite{1}. 

We can define Davenport-Schinzel sequences in a more intuitive way using the idea of sequence pattern avoidance. A sequence $s$ {\it contains} a sequence $u$ if some subsequence of $s$ can be changed into $u$ by a one-to-one renaming of its letters; we say such a subsequence is isomorphic to $u$. If $s$ does not contain $u$, then $s$ {\it avoids} $u$. The sequence $s$ is called {\it r-sparse} if any $r$ consecutive letters in $s$ are pairwise different. Thus Davenport Schinzel sequences of order $s$ correspond to $2$-sparse sequences which avoid an alternation $abab \cdots$ of length $s+2$. 

An important question in pattern avoidance is finding the maximum length of any sequence that avoids a given pattern. If $u$ is a sequence with $r$ distinct letters, then the extremal function $Ex(u, n)$ is the maximum length of any $r$-sparse sequence with $n$ distinct letters that avoids $u$. It is clear that $Ex(u, n) \geq n$ if $u$ has at least one letter that occurs twice. Moreover by the pigeonhole principle, $Ex(u, n) \leq {n \choose r}lr$, where $l$ denotes the length of $u$. Our main goal is to improve the upper bounds and lower bounds on extremal functions so that they are as close as possible.

Maximum lengths of Davenport-Schinzel sequences have been well-studied. If $a$ and $b$ are different letters and $u=abab\cdots$ is an alternation of length $s+2$, then $Ex(u,n)$ is exactly the maximum length of an order $s$ Davenport Schinzel sequence. It is well-known and easy to show that $Ex(a, n)=0, Ex(a b, n)=1, Ex(a b a, n)= n$ and $Ex(a b a b, n)=2n-1$. For alternations $u$ of greater length, $Ex(u, n)$ is not linear in $n$. Nivasch \cite{8} and Klazar \cite{7} proved that $Ex(ababa,n) \sim 2n\alpha(n)$, where $\alpha(n)$ is the extremely slow growing inverse Ackermann Function; we refer the reader to \cite{8} for more information on the inverse Ackermann Function. Agarwal, Sharir, Shor \cite{2} and Nivasch \cite{8} proved that if $u$ is an alternation of length $2t+4$, then $Ex(u, n)= n 2^{\frac{1}{t!}\alpha(n)^{t}\pm O(\alpha(n)^{t-1})}$ for $t\geq 1$.

Besides alternations and Davenport-Schinzel sequences, more general patterns and sequences have also been studied. A {\it generalized Davenport-Schinzel sequence} is an $r$-sparse sequence that does not contain a sequence $u$, where $u$ has $r$ distinct letters (and need not be an alternation). We are interested in the maximum length of a generalized Davenport-Schinzel sequence, which is precisely $Ex(u, n)$.  Fox {\it et al.} \cite{5} and Suk {\it et al.} \cite{10} used bounds on the lengths of generalized Davenport-Schinzel sequences to prove that $k$-quasiplanar graphs on $n$ vertices with no pair of edges intersecting in more than $t$ points have at most $(n\log n)2^{\alpha(n)^{c}}$ edges, where $c$ is a constant depending only on $k$ and $t$. 

General approaches to bounding $Ex(u,n)$ for all patterns $u$ have been found. In \cite{7}, Klazar considered special sequences called formations in order to bound general extremal functions. An $(r, s)$-{\it formation} is a concatenation of $s$ permutations of $r$ distinct letters. Klazar \cite{7} considered the function $F_{r,s}(n)$, which is the maximum length of any $r$-sparse sequence with $n$ distinct letters which avoids all $(r, s)$-formations. The relevance of this function to the extremal function lies in the fact that $Ex(u, n) \le F_{r,s}(n)$ for any sequence $u$ with $r$ distinct letters and length $s$. This inequality is a direct consequence of the fact that every $(r,s)$-formation contains $u$. Nivasch \cite{8} later improved this inequality to $Ex(u, n) \le F_{r,s-r+1}(n)$, for any sequence $u$ with $r$ distinct letters and length $s$.

Much work has been done on $F_{r,s}(n)$. Klazar \cite{6} proved that $F_{r,2}(n)=O(n)$ and $F_{r,3}(n)= O(n)$ for every $r$. Nivasch \cite{8} proved that $F_{r,4}(n)=\Theta(n\alpha(n))$ for $r\geq 2$.  Agarwal, Sharir, Shor \cite{2} and Nivasch \cite{8} proved that $F_{r,s}(n)=n2^{\frac{1}{t!}\alpha(n)^{t}\pm O(\alpha(n)^{t-1})}$ for all $r\geq 2$ and odd $s \geq 5$ with $t= \frac{s-3}{2}$. All of these bounds on $F_{r,s}(n)$ imply corresponding upper bounds on $Ex(u, n)$ by the comments mentioned in the previous paragraph.

In order to obtain the best possible bounds on extremal functions using formations, it is an important problem to find values of $r$ and $s$ for which we can guarantee that $Ex(u,n) \le F_{r,s}(n)$ or $Ex(u,n) = O(F_{r,s}(n))$. To this end, a function called formation width was introduced in \cite{GPT}.  The {\it formation width} of $u$, $\mathit{fw(u)}$, is the minimum value of $s$ such that there exists an $r$ for which every $(r, s)$-formation contains $u$. The {\it formation length} of $u$, $\mathit{fl}(u)$, is the minimum $r$ such that every $(r, fw(u))$-formation contains $u$. The following Lemma relates $fw(u)$ to $Ex(u,n)$.

\begin{lem} \cite{GPT} \label{1.4} $Ex(u, n)=O(F_{\mathit{fl}(u),fw(u)} (n))$ for any sequence $u$. \end{lem}

In view of Lemma \ref{1.4}, computing $fw(u)$ for a sequence $u$ implies an upper bound on $Ex(u,n)$. For instance, if $fw(u) \leq 3$, then applying Lemma \ref{1.4} gives $Ex(u, n) = O(F_{fl(u),fw(u)}(n))=O(n)$, by the results on $F_{r,2}(n)$ and $F_{r,3}(n)$ mentioned above. Every sequence $u$ with $fw(u) \leq 3$ was identified in \cite{GPT} and, as a consequence, these sequences $u$ satisfy $Ex(u,n)=O(n)$ as well.

In this paper, we identify every sequence $u$ that has alternation length $5$ (i.e. $u$ contains $a b a b a$ but not $a b a b a b$) and formation width $4$. Note that for such sequences $u$, we have $Ex(u, n) =O(F_{fl(u),4}(u))=O(n\alpha(n))$ by Lemma \ref{1.4} and the bound on $F_{r,4}(n)$ mentioned above. Since u contains $ababa$, we also have $Ex(u,n)=\Omega(n\alpha(n))$ by the result that $Ex(ababa,n) \sim 2n \alpha(n)$ mentioned above and because of Lemma 1.1b in \cite{6}. Thus every identified sequence of alternation length $5$ and formation width $4$ has a tight bound of $\Theta(n\alpha(n))$ on the extremal function. By using formation width, we have identified the extremal function for infinite classes of previously unidentified sequences.

The significance of this result lies in the fact that $n \alpha(n)$ is nearly the lowest possible order that an extremal function can have. An implication of our result is that we have also identified every sequence with alternation length $5$ for which we may get tight bounds on the extremal function using only formation width and containment of the alternation.

The power of formation width lies in the fact that it is computationally feasible to directly compute formation width of small sequences. In contrast, it is almost impossible to directly compute the extremal function, as it requires mathematical proof to show that it holds for all $n$. In the appendix we include a shorter and faster algorithm than the one included in \cite{GPT} for computing formation width. Thus, our main theorem and our more efficient algorithm highlight the efficacy of formation width for deriving sharp bounds on extremal functions when there are already matching lower bounds.

In Section \ref{prel}, we prove preliminary results. In Section \ref{list}, we identify the sequences with formation width $4$, alternation length $5$, and $n$ distinct letters for $n \geq 6$, and we prove that all of these sequences have formation width $4$ in Section \ref{fw}. In Section \ref{only}, we prove that the sequences from Section \ref{list} are the only sequences with formation width $4$, alternation length $5$, and $n$ distinct letters for $n \geq 6$. In the appendix, we show the code we used to generate the list of sequences for $n \leq 6$.

\section{Preliminary results}\label{prel}

In this section, we make observations about all sequences $u$ which have formation width $4$ and alternation length $5$. These observations will be useful in the proof of our main result.

Let $u'$ be a sequence obtained by deleting a letter that occurs only once in a sequence $u$ with at least two distinct letters. Then $fw(u)=fw(u')$ by Corollary 13 in \cite{GPT} and $u$ has alternation length $5$ if and only if $u'$ does as well. Thus we will only consider those sequences $u$ for which each letter occurs at least twice (we call such a sequence \emph{reduced}), since all other sequences are obtained by adding a finite number of letters, each occuring once, to a reduced sequence.

Furthermore, if a letter occurs at least $4$ times in a reduced sequence $u$ with at least two distinct letters, then $u$ has a subsequence $u'$ on 2 letters with length $6$. Note that $fw(u) \ge fw(u') =5$, where the equality follows from Lemma 17 in \cite{GPT}. Also, if there are two letters $x$ and $y$ that both occur $3$ times in $u$, then the occurences of $x$ and $y$ in $u$ alone form a subsequence $u'$ such that $fw(u) \ge fw(u') = 5$ by Lemma 17 in \cite{GPT}. Thus if $u$ is an $n$-letter reduced sequence such that $fw(u)=4$ and $u$ contains $ababa$, then $u$ must have exactly one letter occuring $3$ times and all other letters occuring twice.

The following lemma is a more complex observation about reduced sequences with formation width $4$ and alternation length $5$.

\begin{lem}
\label{distinct}
If $u$ is a reduced sequence on $n$ letters that has a formation width of $4$ and an alternation length of $5$, then either the first $n$ letters or the last $n$ letters of $u$ must be pairwise distinct.
\end{lem}
\begin{proof}
We proceed by induction on $n$. We used the Python algorithm in the appendix to verify that the lemma is true for all $n \le 6$. Suppose for some $n \geq 7$ that every reduced sequence with $n-1$ distinct letters, formation width $4$, and alternation length $5$ always has the first $n-1$ letters or the last $n-1$ letters pairwise distinct.  Then we prove that that every reduced sequence with $n$ distinct letters, formation width $4$, and alternation length $5$ always has the first $n$ letters or the last $n$ letters pairwise distinct.

Assume for contradiction that there exists an $n$ letter sequence $v$ such that $fw(v) = 4$, $v$ contains $ababa$, and both the first and last $n$ letters of $v$ have at least two occurences of a letter. Let the copy of $ababa$ in $v$ be represented by the letters $x$ and $y$, i.e. $v$ has a subsequence $xyxyx$. Note that this implies all letters except $x$ occur exactly twice in $v$.

If $v$ has a letter besides $x$ or $y$ that occurs once in the first $n$ letters and once in the last $n$ letters, then delete this letter to get a sequence $v'$ that contradicts the inductive hypothesis. Thus, we may assume that all letters other than $x$ and $y$ occur either twice in the first $n$ letters, twice in the last $n$ letters, or in the middle and somewhere else. We consider several cases based on the position of the subsequence $xyxyx$ in $v$.

\noindent
{\bf Case 1: $x$ or $y$ is the middle letter of $v$.}

{\bf Case 1A: The first or third $x$ is the middle letter of $v$.} Since the first $n$ and last $n$ letters both have a letter occuring twice, $v$ has $ccxyxyx$ or $xyxyxcc$ as a subsequence, for some letter $c$. But $fw(v) \ge fw(ccxyxyx)=fw(xyxyxcc)=fw(xyxyx)+1=5$ by Lemma 5 and Corollary 13 in \cite{GPT}, contradicting the assumption that $fw(v)=4$.

{\bf Case 1B: The second $x$ is the middle letter of $v$.} Then all letters besides $x$ or $y$ must occur twice in the first $n$ or twice in the last $n$ letters. In the first $n$ letters, delete two occurences of any letter other than $x$ and $y$ to get a new sequence $v'$ on $n-1$ letters. In $v'$, $x$ occurs twice in the first $n-1$ letters and some letter $c$ occurs twice in the last $n-1$ letters, where $c$ is a letter other than $x$, $y$, or the middle letter of $v'$. Therefore $v'$ contradicts the inductive hypothesis.

{\bf Case 1C: $y$ is the middle letter of $v$.} Without loss of generality, assume that the first $y$ is the middle letter of $v$. Then delete both occurences of a letter besides $x$ in the first $n$ letters of $v$ to obtain $v'$. Then both the first $n-1$ letters and the last $n-1$ letters of $v'$ have two occurrences of a letter besides $x$ or $y$, which contradicts the inductive hypothesis.

\noindent
{\bf Case 2: Neither $x$ nor $y$ is the middle letter of $v$.}

Let $t$ be the middle letter.

{\bf Case 2A: $xyxyx$ is a subsequence of the first $n$ letters or the last $n$ letters of $v$.} This is a contradiction for the same reason as Case $1A$.

{\bf Case 2B: $xyxy$ is a subsequence of the first $n$ letters and $x$ occurs in the last $n$ letters of $v$.} Let $v'$ be a sequence obtained by deleting a letter besides $x$, $y$, or $t$ that occurs twice in the last $n$ letters. Then the first $n-1$ letters of $v'$ have two occurrences of $x$ and the last $n-1$ letters of $v'$ must have another letter occuring twice, contradicting the inductive hypothesis.

{\bf Case 2C: $xyx$  is a subsequence of the first $n$ letters and $yx$ is a subsequence of the last $n$ letters of $v$.} Let $v'$ be a sequence obtained by deleting a letter besides $x$, $y$, or $t$ that occurs twice in the last $n$ letters. Then the first $n-1$ letters of $v'$ have two occurrences of a letter other than $x$, $y$, or $t$, as do the last $n-1$ letters of $v'$. Thus $v'$ contradicts the inductive hypothesis.

We have shown that every case leads to a contradiction. Thus, our induction is complete.
\end{proof}

Given Lemma \ref{distinct}, when we identify the sequences $u$ with $n$ distinct letters for which $fw(u) = 4$ and $u$ contains $ababa$, we will only consider the sequences $u$ where the first $n$ letters are all distinct; the sequences in which the last $n$ letters are distinct can be obtained by reversing a sequence in which the first $n$ letters are distinct. We conclude this section with a final observation, also proved by induction.

\begin{lem}\label{midd}
Let $n\ge 6$. If $u$ is an $n$-letter reduced sequence with formation width $4$ and alternation length $5$ such that the first $n$ letters of $u$ are distinct, then the middle letter of $u$ must always be the same as the first or second letter of $u$. 
\end{lem}
\begin{proof}
We prove the claim by induction. For the case $n=6$, see the list in the appendix. For the inductive hypothesis, assume that for some $n \geq 7$ the middle letter is the same as the first or second letter in all $(n-1)$-letter sequences $u$, of formation width $4$ and alternation length $5$, such that the first $n-1$ letters of $u$ are distinct. We prove the same is true when $n-1$ is replaced by $n$.

Suppose for contradiction that there exists a sequence $v$ on $n$ distinct letters such that the first $n$ letters of $v$ are distinct, $fw(v)=4$, $v$ has a subsequence $xyxyx$, and $v$ has a middle letter $t$ that is not the same as the first or second letters of $v$. Then let $v'$ be the sequence obtained by deleting the two occurences of some letter other than $x$, $y$, $t$, the first, or the second letter of $v$. The deleted letter had to occur both in the first $n$ and the last $n$ letters of $v$, so $v'$ still has a middle letter that is not its first or second letter. Thus $v'$ contradicts the inductive hypothesis.

Therefore our induction is complete.
\end{proof}

\section{Proof of Main Theorem}\label{list}

In this section, we state and prove our main theorem. Throughout the rest of the paper, we number the letters of sequences $1, 2, \ldots$ in order of their first occurrence in the sequence.

\begin{theo}\label{mainlist}

Up to reversal and adding a finite number of distinct letters that each occur once, every sequence that has formation width $4$ and alternation length $5$ must be isomorphic to one of the following sequences:

\begin{scriptsize}
\begin{itemize}

\item $12121$

\item $1233121$

\item $123412134$

\item $123441213$

\item $123413214$

\item $123431243$

\item $123421432$

\item $123431214$

\item $123432143$

\item $123412143$

\item $12345124325$

\item $12345312154$

\item $1 2 \ldots n 1 3 \ldots i 2 (i+1) \ldots n 1$ for $n \geq 4$ and $i=3, 4, \ldots n-1$

\item $1 2 \ldots n 1 2 \ldots (i-1) (i+1) \ldots n i 1$ for $n \geq 4$ and $i = 3, 4, \ldots, n-1$

\item $1 2 \ldots n 1 3 \ldots n 2 1$ for $n \geq 3$

\item $1 \ldots n 2 \ldots n 21$ for $n \geq 3$

\item $1 \ldots n 2 1 3 \ldots n 1$ for $n \geq 3$

\item $1 \ldots n 2 1  3 \ldots n 2$ for $n \geq 3$

\item $1\ldots n1\ldots ni$, for $n \geq 2$ and $i  = 1,...,n-1$

\item $1\ldots n1 \ldots (n-1)i n$ for $n \geq 3$ and $i = 1,...,n-2$

\item $1 \ldots n 1 2 4 \ldots n 3 2$ for $n \geq 4$

\item $1 \ldots n 1 3 \ldots n 3 2$ for $n \geq 4$

\end{itemize}
\end{scriptsize}
\end{theo}

\begin{corollary} If $u$ is a sequence that is listed in Theorem \ref{mainlist}, then $Ex(u,n)=\Theta(n \alpha(n))$.
\end{corollary}

Clearly all of the above sequences have alternation length $5$. In \ref{fw}, we first prove that each of these sequences has formation width $4$, and in \ref{only}, we show that these are indeed the only reduced sequences (up to isomorphism and reversal) that have alternation length $5$ and formation width $4$.

\subsection{Proof that the sequences have formation width $4$} \label{fw}

Using the code for formation width in the appendix, we have verified that every sequence in Theorem \ref{mainlist} with $6$ or fewer letters indeed has formation width $4$. Thus, we just focus on showing that the general classes listed in Theorem \ref{mainlist} always have formation width $4$.

For every sequence $u$ in Theorem \ref{mainlist}, we have $fw(u) \ge fw(ababa)=4$. Thus we just have to show that $fw(u) \le 4$. Call a formation \emph{binary} if each of its permutations is the same or the reverse of the first. The following result about binary formations was proved in \cite{GPT}.

\begin{lem}\label{binary}
If $u$ has $r$ distinct letters, then every binary $(r, s)$-formation contains $u$ if and only if $s \geq fw(u)$. 
\end{lem}

In view of Lemma \ref{binary}, to show the sequences above have formation width at most $4$, it suffices to show that each of them are contained in every binary $(n, 4)$-formation. Let the first permutation of every formation be $p$, and let its reverse be $\bar{p}$. In our proofs, we just have to show that all $8$ possibilities for the binary formation (i.e. $pppp, ppp\bar{p}, pp\bar{p}p, p\bar{p}pp, pp\bar{p}\bar{p}, p\bar{p}p\bar{p}, p\bar{p}\bar{p}p, p\bar{p}\bar{p}\bar{p}$) contain $u$. In each case, we show that we can number the letters of $p$ on $1, 2, ..., n$ in some way so that the formation has $u$ as a subsequence.

Lemma 28 in \cite{GPT} proved that $fw(1 2 \ldots n13 \ldots i 2 (i+1) \ldots n 1)=4$ for $i=3, 4, \ldots n-1$, $fw(1 2 \ldots n 1 2 \ldots (i-1) (i+1) \ldots n i 1) = 4$ for $i=3, 4, \ldots n-1$, and $fw(1 2 \ldots n 1 3 \ldots n 2 1) = 4$. We show in the following lemmas that the rest of the sequences in Theorem \ref{mainlist} must also have formation width $4$.

\begin{lem}
$fw(1 \ldots n 2 \ldots n 21)=4$
\end{lem}
\begin{proof}

\noindent
{\bf Case 1: The first two permutations are $pp$.} Let $p=1 \ldots n$: take $2$ in the third and $1$ in the last permutation.

\noindent
{\bf Case 2: The formation is $p\bar{p}pp$.} Let $p=3 \ldots n 2 1 $:  take $12$ in the second, $3 \ldots n 2$ in the third, and $3 \ldots n 2 1$ in the last.

\noindent
{\bf Case 3: The formation is $p \bar{p} \bar{p} \bar{p}$.} Let $p=1 2 n \ldots 3$: take $12$ in the first, $3 \ldots n 2$ in the second, and $3 \ldots n 2 1$ in the last.

\noindent
{\bf Case 4: The formation is $p\bar{p}\bar{p}p$.} Let $p=n \ldots 1$: take $1 \ldots n$ in the second, $2 \ldots n$ in the third, and $21$ in the last.

\noindent
{\bf Case 5: The formation is $p\bar{p}p\bar{p}$.} Let $p=1 \ldots n$: take $1 \ldots n$ in the first, $2 \ldots n$ in the third, and $21$ in the last.
\end{proof}

\begin{lem}
$fw(1 \ldots n 2 1 3 \ldots n 1)=4$
\end{lem}
\begin{proof}

\noindent
{\bf Case 1: The last two permutations are the same, or the formation contains $p\bar{p}\bar{p}$ or $\bar{p}p p$.} Let the repeated permutation be $3 \ldots n 21$.

\noindent
{\bf Case 2: The formation is $p p p \bar{p}$.} Let $p=1 \ldots n$.

\noindent
{\bf Case 3: The formation is $p p \bar{p}p$.} Let $p=3 \ldots n 1 2$.
\end{proof}

\begin{lem}
$fw(1\ldots n1\ldots ni)=4$ for $i=1,\ldots ,n-1$
\end{lem}
\begin{proof} Two of the first $3$ permutations in the formation must be the same. Thus they contain $1\ldots n1\ldots n$. We can choose $i$ in the fourth permutation.
\end{proof}

\begin{lem}
$fw(1\ldots n1 \ldots (n-1)i n)=4$ for $i=1,\ldots ,n-2$
\end{lem}
\begin{proof}

\noindent
{\bf Case 1: The formation has $3$ permutations the same.} The first two of the three permutations contain $1\ldots n1 \ldots (n-1)$ and the last contains $i n$

\noindent
{\bf Case 2: The formation is $pp \bar{p} \bar{p}$.} Let $p=1 \ldots n$. We can choose the $i$ in the third permutation and the $n$ in the fourth.

\noindent
{\bf Case 3: The formation is $p\bar{p}\bar{p}p$.} Let $p=n-1 \ldots 1 n$.

\noindent
{\bf Case 4: The formation is $p\bar{p}p\bar{p}$.} Let $p=n 1 \dots n-1$.
\end{proof}

\begin{corollary}
$fw(1 \ldots n 2 1  3 \ldots n 2)=4$
\end{corollary}
\begin{proof} This is the reverse of $1 \ldots n 1  \ldots (n-1) 1 n$, which is of the form $1\ldots n1 \ldots (n-1)i n$.
\end{proof}

\begin{lem}
$fw(1 \ldots n 1 2 4 \ldots n 3 2)=4$
\end{lem}
\begin{proof}

\noindent
{\bf Case 1: The first $2$ permutations are the same, or the formation contains $pp\bar{p}$ or $\bar{p}\bar{p}p$.} Let the repeated permutation be $1 \ldots n$.

\noindent
{\bf Case 2: The formation is $p \bar{p} \bar{p} \bar{p}$.} Let $p=1 n \ldots 4 2 3$.

\noindent
{\bf Case 3: The formation is $p\bar{p}pp$.} Let $p=4 \ldots n 3 1 2$.
\end{proof}

\begin{lem}
$fw(1 \ldots n 1 3 \ldots n 3 2)=4$
\end{lem}
\begin{proof}

\noindent
{\bf Case 1: The first $2$ permutations are the same, or the formation contains $pp\bar{p}$ or $\bar{p}\bar{p}p$.} Let the repeated permutation be $1 \ldots n$.

\noindent
{\bf Case 2: The formation is $p \bar{p} \bar{p} \bar{p}$.} Let $p=1 n \ldots 4 2 3$.

\noindent
{\bf Case 3: The formation is $p\bar{p}pp$.} Let $p=4 \ldots n 1 3 2$.
\end{proof}

Thus we have shown all sequences in Theorem \ref{mainlist} indeed have formation width $4$ (and alternation length 5). In the next section, we prove that these are the only such sequences with formation width $4$ and alternation length $5$.

\subsection{Proof that the sequences are the only sequences with formation width $4$ and alternation length $5$}\label{only}

In this section we will show that the sequences $u$ from Section \ref{list} are the only reduced sequences up to isomorphism and reversal such that $fw(u)=4$ and $u$ contains $ababa$. By the list in the appendix, we have verified that Theorem \ref{mainlist} contains all sequences on at most $6$ letters that have formation width 4 and alternation length 5. Thus we just need to show that all sequences with at least $6$ letters that have formation width 4 and alternation length $5$ must be equivalent to one of the general classes in Theorem \ref{mainlist}. In order to do this, we will split the proof into cases for all possible sequences $u$. 

By the observations in Section \ref{prel}, we may suppose that $u$ is reduced, $fw(u)=4$, $u$ contains $ababa$, and the first $n$ letters of $u$ are $1 \ldots n$. We first identify every sequence $u$ that ends in $i$ for $i = 3,\ldots, n-1$. Next we identify every sequence $u$ that ends in $n$. This leaves only the sequences $u$ that end in $1$ or $2$. 

The sequences that end in $1$ and have middle letter $1$ were identified in \cite{GPT}. We show that every sequence $u$ ending in $1$ with middle letter $2$ has second to last letter $2$ or $n$, and then we identify all such sequences. 

Next we identify every sequence $u$ ending in $2$ with middle letter $2$. Then we show that if $u$ has middle letter $1$ and last letter $2$, then the letter to the right of the middle of $u$ must be $2$ or $3$, and we identify all such sequences. This covers every possible case by Lemmas \ref{distinct} and \ref{midd}.

Each of the following lemmas either categorizes the sequences $u$ or narrows the possibilities for such sequences. We will prove each lemma by induction, using the list of sequences of length $6$ in the appendix for the base case ($n=6$ letters). 

In the proofs of each of the following lemmas that identify a specific sequence, we suppose for contradiction that $n$ is minimal so that there exists a sequence $u$ with $n > 6$ letters that does not have the form of the sequence $v$ identified in the lemma statement. For each such $u$ and $v$, define $z$ and $j$ to be the letters in $u$ and $v$ respectively in the first location where $u$ and $v$ have different letters. This means the letters before $z$ in $u$ must agree with the letters before $j$ in $v$. 

For each lemma, the inductive hypothesis is that the lemma is true for the case when $u$ has $n-1$ distinct letters. Moreover, without loss of generality suppose that $u$ has the subsequence $x y x y x$. This means that $x$ occurs $3$ times and all other letters occur $2$ times in $u$. 

\begin{lem}
If $u$ ends in $i$, for $i = 3,\ldots, n-1$, then $u = 1\ldots n 1 \ldots n i$.
\end{lem}
\begin{proof}
Suppose that $u$ has $n > 6$ letters and $u$ is not of the form $1\ldots n 1 \ldots n i$. We may delete both occurrences of any letter besides $1, 2, x, y, z, i, j, n$ to obtain a sequence that contradicts the inductive hypothesis. Since $n$ may be as low as $7$, we will show that some of these letters are the same.

First, we will show by another induction that $i=x$. The case $n = 6$ follows from the list of sequences in the appendix. If $u$ is a sequence with $n > 6$ distinct letters such that $i \neq x$, then we may delete both occurrences of any letter not equal to $1, 2, x, y, i,n$ to contradict the inductive hypothesis.

If $i=n-1$, then $x=n-1$. Since $u$ contains $xyxyx$ and $n$ is the only letter besides $x$ that appears twice after the first $n-1$ letters, $y=n$. Since $x=i=n-1$ and $y=n$, we may delete any letter besides $1,2,x,y,j,z$. Otherwise $3 \le i \le n-2$, and since $x=i$, we may delete any letter besides $1,2,x,y,j,z$. Both of the resulting sequences contradict the inductive hypothesis.
\end{proof}

Next we categorize all sequences that end in $n$.

\begin{lem}
If $u$ ends in $n$, then $u=1\ldots n 1 \ldots (n-1) i n$, where $i$ may be $1,\ldots, n-2$.
\end{lem}
\begin{proof}
Suppose that $u$ has $n > 6$ letters and $u$ is not of the form $1\ldots n 1 \ldots (n-1) i n$. Then we may delete any letter that is not $n,n-1,j,z,x,y$ to get a sequence that contradicts the inductive hypothesis.
\end{proof}

All that remains is to categorize all sequences satisfying the conditions and ending with $1$ or $2$. In the next four lemmas, we first categorize all sequences ending with $1$. Note that the first of the next four lemmas follows directly from Lemmas 28 and 31 in \cite{GPT}.

\begin{lem}\cite{GPT}
If $u$ ends in $1$ and has middle letter $1$, then $u=1 2 \ldots n 1 v 1 $, where $v$ is a permutation of $2 \ldots n$ obtained by either moving $2$ anywhere else in $2 \ldots n$ or moving any letter in $2 \ldots n$ to the end of $2 \ldots n$. Note that this includes $v=2 \ldots n$.
\end{lem}

\begin{lem}
If $u$ has last letter $1$ and middle letter $2$, then the second to last letter of $u$ must be $2$ or $n$.
\end{lem}
\begin{proof}
The list of sequences in the appendix shows that this lemma is true for $n = 6$. Suppose that $u$ has $n > 6$ letters and has second to last letter $t$ which is not $2$ or $n$. Then we can delete any letter not $1,2,x,y,n,t$ to contradict the inductive hypothesis.
\end{proof}

\begin{lem}
If $u$ has last letter $1$, middle letter $2$, and second to last letter $2$, then $u=1 \ldots n 2 \ldots n 2 1$.
\end{lem}
\begin{proof}
The case $n = 6$ can be verified with the list in the appendix. Suppose that $u$ has $n > 6$ letters and $u$ is not of the form $1 \ldots n 2 \ldots n 2 1$. Since $x=2$, we can delete any letter not $1,j,z,x,y$ to get a contradiction.
\end{proof}

\begin{lem}
If $u$ has last letter $1$, middle letter $2$, and second to last letter $n$, then $u=1 \ldots n 2 1 3 \ldots n 1$.
\end{lem}
\begin{proof}
The case $n = 6$ can be verified with the list in the appendix. Suppose that $u$ has $n > 6$ letters and $u$ is not of the form $1 \ldots n 2 1 3 \ldots n 1$. Then we can delete any letter not $1,2,j,z,x,y,n$ to contradict the inductive hypothesis. Since $n$ can be as low as $7$, it will suffice to show that two of these letters are the same.

We show by induction that $x=1$, i.e. $1$ must occur $3$ times in $u$. The case of $n=6$ is true from the list in the appendix. If $u$ is a sequence with $n > 6$ distinct letters such that $x \neq 1$, then we may delete a letter not equal to $1,2,n,x,y$ to get a contradiction.
\end{proof}

Now we classify the sequences ending in $2$.

\begin{lem}
If $u$ has last and middle letter $2$, then $u=1 \ldots n 2 1 3 \ldots n 2$.
\end{lem}
\begin{proof}
The case $n = 6$ can be verified with the list in the appendix. Suppose that $u$ has $n > 6$ letters and $u$ is not of the form $1 \ldots n 2 1 3 \ldots n 2$. Since $x = 2$, we can delete any letter not $1,j,z,x,y$ to get a contradiction.
\end{proof}

\begin{lem}
If $u$ has middle letter $1$ and last letter $2$, then the letter to the right of the middle of $u$ must be $2$ or $3$.
\end{lem}
\begin{proof}
The case $n = 6$ can be verified with the list in the appendix. Suppose that $u$ has $n > 6$ letters and the letter $t$ to the right of the middle of $u$ is not $2$ or $3$. Then we can delete any letter not $1,2,3,t,x,y$ to get a contradiction.
\end{proof}

\begin{lem}
If $u$ has last letter $2$, $1$ in the middle, and $3$ right after the middle $1$, then $u=1 \ldots n 1 3 \ldots n 3 2$.
\end{lem}
\begin{proof}
The case $n = 6$ can be verified with the list in the appendix. Suppose that $u$ has $n > 6$ letters and $u$ is not of the form $1 \ldots n 1 3 \ldots n 3 2$. Then we can delete any letter not $1,2,3,z,x,y,j$ to contradict the inductive hypothesis. Since $n$ can be as low as $7$, it will suffice to show that two of these letters are the same.

We show by induction that $x=3$. The case of $n=6$ is true from the list in the appendix. If $u$ is a sequence with $n > 6$ distinct letters such that $x \neq 3$, then we may delete a letter not equal to $1,2,3,x,y$ to get a contradiction. 
\end{proof}

\begin{lem}
If $u$ has last letter $2$, $1$ in the middle, and $2$ right after the middle $1$, then $u=1 \ldots n 1 2 4 \ldots n 3 2$.
\end{lem}
\begin{proof}
The case $n = 6$ can be verified with the list in the appendix. Suppose that $u$ has $n > 6$ letters and $u$ is not of the form $1 \ldots n 1 2 4 \ldots n 3 2$. Since $x = 2$, we can delete any letter not $1,x,y,z,j$ to get a contradiction.
\end{proof}

The lemmas above have covered every possible case. Therefore, up to reversal and isomorphism, the sequences in Theorem \ref{mainlist} are indeed the only sequences of formation width $4$ and alternation length $5$.

\subsection*{Acknowledgments}
The authors thank the Program for Research in Mathematics, Engineering and Science (PRIMES) and the Research Science Institute (RSI) at MIT for supporting this research. Jesse Geneson was supported by the NSF Graduate Research Fellowship under Grant No. 1122374.

\appendix
\section{Algorithm for computing $fw$}

Below is the Python code used to generate the list in the next section. If $u$ is a sequence with $r$ distinct letters, then the formation width function increments $s$ starting from $1$ until it finds that every binary $(r,s)$-formation contains $u$. If some binary $(r, s)$-formation $f$ contains $u$, then for every $s' > s$ the algorithm does not check for containment of $u$ in any binary $(r, s')$-formations $f'$ for which $f'$ restricted to its first $s$ permutations is equal to $f$. The formation width function below runs faster than the function in \cite{GPT}. Comments are added before each section of code.

\begin{scriptsize}
\begin{verbatim}
from itertools import permutations
from sets import Set
\end{verbatim}
\end{scriptsize}

\noindent determines whether one sequence is a subsequence of another: 

\begin{scriptsize}
\begin{verbatim}
def issubseq(seq, subseq): 
    if len(subseq) == 0:
        return True
    else:
        if len(seq) == 0:
            return False
        elif seq[-1] == subseq[-1]:
            return issubseq(seq[:-1],subseq[:-1])
        elif seq[-1] != subseq[-1]:
            return issubseq(seq[:-1],subseq)
\end{verbatim}
\end{scriptsize}
  
\noindent determines the formation width of u: 

\begin{scriptsize}
\begin{verbatim}
def fw(u): 
    l=len(set(u))     
    v = list(u) 
    rsformset = set() 
    rsformset1 = set()     
    q = tuple(range(l)) 
    q1 = q[::-1] 
    rsformset.add(q) 
    rsform1=q
    while len(rsformset)!=0: 
        for rsforms in rsformset: 
            done=False
            for perms in permutations(range(l)): 
                for i in range(len(u)):  
                    v[i] = perms[u[i]] 
                if issubseq(rsforms, v):
                    done=True
                    break
            if not done:
                rsformset1.add(rsforms+q) 
                rsformset1.add(rsforms+q1) 
                rsform1=rsforms+q 
        rsformset.clear() 
        for rsform in rsformset1: 
            rsformset.add(rsform) 
        rsformset1.clear() 
    return len(rsform1)//l
\end{verbatim}
\end{scriptsize}

\noindent outputs the index of the first occurrence of a letter in a sequence:

\begin{scriptsize}
\begin{verbatim}
def fstocc(x,i):
    for t in range(len(x)):
        if x[t] == i:
            return t
\end{verbatim}
\end{scriptsize}
	
\noindent outputs the set of sequences with 2 occurrences of each letter such that letters are 0,1,..,n-1 and first occurrences of letters are in increasing order:

\begin{scriptsize}
\begin{verbatim}
def letocc2x(n):
    final = set()
    if n == 1:
        final.add((0,0))
    else:
        for s in letocc2x(n-1):
            for i in range(fstocc(s,n-2)+1,len(s)+1):
                t = list(s)
                t.insert(i, n-1)
                r1 = tuple(t)
                for j in range(fstocc(r1,n-1)+1,len(r1)+1):
                    t = list(r1)
                    t.insert(j, n-1)
                    r2 = tuple(t)
                    final.add(r2)			
    return final
\end{verbatim}
\end{scriptsize}

\noindent outputs the set of sequences that contain ababa with 3 occurrences of one letter and 2 occurrences of every other letter such that letters are 0,1,..,n-1 and first occurrences of letters are in increasing order:

\begin{scriptsize}
\begin{verbatim}
def a3xotherlet2x(n):
    start = letocc2x(n)
    final = set()
    for x in start:
        for i in range(n):
            if i == 0:
                for j in range(1,len(x)+1):
                    t = list(x)
                    t.insert(j,i)
                    for t1 in range(n):
                        for t2 in range(t1+1,n):
                            if (issubseq(tuple(t),(t1,t2,t1,t2,t1)) or issubseq(tuple(t),(t2,t1,t2,t1,t2))):
                                final.add(tuple(t))
            else:
                for j in range(fstocc(x,i-1)+1,len(x)+1):	
                    t = list(x)
                    t.insert(j,i)
                    for t1 in range(n):
                        for t2 in range(t1+1,n):
                            if (issubseq(tuple(t),(t1,t2,t1,t2,t1)) or issubseq(tuple(t),(t2,t1,t2,t1,t2))):
                                final.add(tuple(t))
    return final
\end{verbatim}
\end{scriptsize}

\noindent outputs every sequence u from a3xotherlet2x(n), n = 2, 3, 4, 5, 6, for which fw(u) = 4; also translates alphabet so that letters are 1,2,..,n, and first occurrences of letters are in increasing order:

\begin{scriptsize}
\begin{verbatim}
for j in range(2, 7):
    for seq in a3xotherlet2x(j):
        if fw(seq) == 4:
            t = list()
            for i in range(len(seq)):
                t.append(str(int(seq[i])+1))
            print "".join(t)
    print ""
\end{verbatim}
\end{scriptsize}

\noindent The program above ran on a MacBook Air with operating system Mavericks version 10.9.4, 1.8 GHz Intel Core i5 processor and 8 GB 1600 MHz DDR3 SDRAM. The program finished outputting the list in the next section in under 10 hours.

\section{The sequences on $n \leq 6$ distinct letters that have formation width $4$ and alternation length $5$}

Every reduced sequence on $n \leq 6$ distinct letters that has formation width $4$ and alternation length $5$ must be isomorphic to one of the following sequences:

\begin{scriptsize}
\begin{verbatim}
12121

1231213
1233121
1213231
1213321
1232132
1232131
1213213
1231232
1231231
1232321
1231321

123421431
123412432
123412134
123142341
123413421
123412341
123243214
123143214
123441213
123412431
123413214
123431243
123423421
123421342
123412343
123241432
123244132
123413241
123142314
123412314
123421432
123431214
123413432
123432143
121342134
123421341
123412342
123412324
123143241
123412143
123241324

12345123454
12345134251
12343521543
12345123415
12341523415
12345124532
12345234521
12345123425
12345123453
12342534215
12342514325
12345124531
12345312154
12345213451
12345123452
12345213452
12342513425
12341534215
12345123541
12341523451
12345134532
12345124325
12345123451
12324513245
12314523145
12345123435
12134521345
12345134521
12345132451

1234561234526
1234562134562
1234561234536
1234561234565
1234561342561
1234561234564
1234561234516
1234561245631
1232456132456
1234561234562
1234562345621
1234562134561
1234516234561
1234256134256
1234561234563
1234561234651
1234156234156
1234561324561
1234526345216
1213456213456
1234561245632
1234561235641
1234516345216
1234561345632
1234561345621
1234526134526
1234561234561
1234561345261
1231456231456
1234516234516
1234561234546
\end{verbatim}
\end{scriptsize}
\end{document}